\documentclass[fleqn,10pt]{wlscirep}

\usepackage{amsthm}
\usepackage{amsmath}
\usepackage{graphicx}

\newtheorem{lemma}{Lemma}

\newcommand{\firstrunpaper}{Hensen \emph{et.al.}\cite{Hensen2015LoopholeFree}}
\newcommand{\pvalue}{$P$-value}
\newcommand{\pvalues}{$P$-values}
\newcommand{\TT}{attempt}
\newcommand{\TTs}{attempts}
\newcommand{\h}{h}
\newcommand{\HH}{H}
\newcommand{\HHTot}{H}
\newcommand{\optTilde}[1]{#1}
\newcommand{\NNN}{{m}}
\newcommand{\N}{m}
\newcommand{\ellInd}{i}

\newcommand{\LHVMs}{\textrm{LHVM} }
\newcommand{\LHVM}{\textrm{LHVM}}

\newcommand{\indep}{\rotatebox[origin=c]{90}{$\models$}}
\newcommand{\pr}[1]{\textrm{\textnormal{Pr}}\left(#1\right)}

\newcommand{\E}[1]{\mathbb E\left[#1\right]}
\newcommand{\Fx}{Q^A}
\newcommand{\Fy}{Q^B}
\newcommand{\Bx}{F^A}
\newcommand{\By}{F^B}
\newcommand{\bellwin}{\beta_{\rm win}}
\def\p{\textrm{Pr}}

\title{Loophole-free Bell test using electron spins in diamond: second experiment and additional analysis}
\author[1,2]{B.~Hensen}
\author[1,2]{N.~Kalb}
\author[1,2]{M.S.~Blok}
\author[1,2]{A.E.~Dr\'{e}au}
\author[1,2]{A.~Reiserer}
\author[1,2]{R.F.L.~Vermeulen}
\author[1,2]{R.N.~Schouten}

\author[3]{M.~Markham}
\author[3]{D.J.~Twitchen}

\author[1]{K.~Goodenough}
\author[1]{D.~Elkouss}
\author[1]{S.~Wehner}
\author[1,2]{T.H.~Taminiau}
\author[1,2,*]{R.~Hanson}

\affil[1]{QuTech, Delft University of Technology, P.O. Box 5046, 2600 GA Delft, The Netherlands}
\affil[2]{Kavli Institute of Nanoscience Delft, Delft University of Technology, P.O. Box 5046, 2600 GA Delft, The Netherlands}
\affil[3]{Element Six Innovation, Fermi Avenue, Harwell Oxford, Didcot, Oxfordshire OX110QR, United Kingdom}
\affil[*]{r.hanson@tudelft.nl}

\begin{abstract}
The recently reported violation of a Bell inequality using entangled electronic spins in diamonds (Hensen \emph{et. al.}, \emph{Nature} 526, 682-686) provided the first loophole-free evidence against local-realist theories of nature. Here we report on data from a second Bell experiment using the same experimental setup with minor modifications. We find a violation of the CHSH-Bell inequality of $2.35 \pm 0.18$, in agreement with the first run, yielding an overall value of $S = 2.38 \pm 0.14$. We calculate the resulting $P$-values of the second experiment and of the combined Bell tests. We provide an additional analysis of the distribution of settings choices recorded during the two tests, finding that the observed distributions are consistent with uniform settings for both tests. Finally, we analytically study the effect of particular models of random number generator (RNG) imperfection on our hypothesis test. We find that the winning probability per trial in the CHSH game can be bounded knowing only the mean of the RNG bias, implying that our experimental result is robust for any model underlying the estimated average RNG bias.
\end{abstract}
\begin{document}

\flushbottom
\maketitle

\thispagestyle{empty}

\section*{Introduction}
Ever since its inception, the counterintuitive predictions of quantum theory have stimulated debate about the fundamental nature of reality. In 1964, John Bell found that the correlations between outcomes of distant measurements allowed under local realism\cite{Einstein1935Can} are strictly bounded, while certain quantum mechanical states are predicted to violate this bound\cite{Bell1964On}. Numerous violations of a Bell inequality in agreement with quantum theory have been reported\cite{Freedman1972Experimental,Aspect1982Experimental,Tittel1998Violation,Weihs1998Violation,Rowe2001Experimental,Matsukevich2008Bell,Ansmann2009Violation,Scheidl2010Violation,Hofmann2012Heralded,Pfaff2013Demonstration,Giustina2013Bell,Christensen2013DetectionLoopholeFree,Ballance2015Hybrid,Dehollain2016BellS}. However, due to experimental limitations additional assumptions were required in all experiments up to 2015 in order to reject the local-realist hypothesis, resulting in loopholes. Last year we reported the first experimental loophole-free violation of the CHSH-Bell inequality using entangled electron spins associated with nitrogen-vacancy (NV) centers in diamond, separated by 1.3~km~\cite{Hensen2015LoopholeFree}. Less than three months after our experiment, two groups observed violations of the CH-Eberhard inequality on spatially-separated photons~\cite{Giustina2015SignificantLoopholeFree,Shalm2015Strong} and before the end of the year first signatures of a CHSH-Bell violation on single Rubidium atoms were found\cite{Weinfurter2016Bell}.

Below, we report on data from a second loophole-free Bell test performed with the same setup as in \firstrunpaper. Additionally, we analyse in detail the recorded distribution of settings choices in both the first and second datasets. Finally, we investigate the effect of arbitrary models underlying the bias in the random number generation.

\section*{Second run}
\label{sec:second_run}
After finishing the first loophole-free Bell experiment in July 2015, both the A(lice) and B(ob) setups were modified and used in various local experiments. In December 2015, we rebuilt the Bell setup for performing a second run of the Bell test, with three small modifications compared to the first run: 

First, we add a source of classical random numbers for the input choices. A random basis choice is now made by applying an XOR operation between a quantum random bit generated as previously~\cite{Abellan2014UltraFast,Mitchell2015Strong,Abellan2015Generation} and classical random bits based on Twitter messages, as proposed by Pironio~\cite{Pironio2015Random}. In particular, we generate two sets of classical random numbers, one for the basis choice at A, and one for the basis choice at B (see details in the following sections). At each location, 8~of these bits are fed into an FPGA. Just before the random basis rotation, the 8~Twitter bits and 1~quantum random bit are combined by subsequent XOR operations. The resulting bit is used as the input of the same microwave switch as used in the first run\cite{Hensen2015LoopholeFree}. The XOR operation takes 70~ns of additional time, shifting the start of the readout pulse to a later time by the same amount. We leave the end of the readout window unchanged, resulting in the same locality conditions as in the first test.

We note that the Twitter-based classical random bits by themselves cannot close the locality loophole: the raw data is available on the Internet well before the trials and the protocol to derive the bits is deterministic and programmed locally. The only operations that are performed in a space-like separated manner are the XOR operations between 8 stored bits. Therefore, strictly speaking only the quantum-RNG is providing fresh random bits. Since a loophole-free Bell test is described solely by the random input bit generation and the outcome recording at A and B (and in our case the event-ready signal recording at C), the second run can test the same null hypothesis as the first run as these events are unchanged. That being said, the use of the Twitter-based classical randomness puts an additional constraint on local-hidden-variable models attempting to explain our data.

Second, we set larger (i.e. less conservative) heralding windows at the event-ready detector in order to increase the data rate compared to the first experiment. We start the heralding window about 700~picoseconds earlier, motivated by the data from the first test (grey line in Fig.~\ref{fig:bell_filter_dependence}). We predefine a window start of 5426.0~ns after the sync pulse for channel~0, and 5425.1~ns for channel 1. We set a window length of 50~ns.

Finally, we also use the $\psi^+$- Bell state, which is heralded by two photo-detection events in the \emph{same} beamsplitter output arm at the event-ready station. In general the fidelity of this Bell state is lower than that of $\psi^-$ due to detector after-pulsing \cite{Bernien2013Heralded} (note that for $\psi^-$ the after-pulsing is not relevant because $\psi^-$ is heralded by photo detection events in \emph{different} beamsplitter output arms). However, we found the after-pulsing effect to be small enough for the detectors used in this run. We set an adapted window length of the second window of 4~ns and 2.5~ns for channels~0,~1 respectively, where the exponentially decaying NV emission is still large relative to the after-pulsing probability. As described below, we can combine the $\psi^-$-related and $\psi^+$-related Bell trials into a single hypothesis test\cite{Elkouss2015Nearly}.

Apart from these modifications, all settings, analysis software, calibrations and stabilisation routines were identical to those in the first run\cite{Hensen2015LoopholeFree}.

\subsection*{Random numbers from Twitter}
\label{sec:bell_twitter_rnd}
After each potential heralding event (corresponding to the $E$-events described in the Supplementary Information of \firstrunpaper), both at location A and B we take 8 new bits from a predefined random dataset (one for A and one for B) based on Twitter messages, to send to the FPGA-based random-number combiner.

The random dataset for A was obtained by collecting 139952 messages from Twitter trending topic with hash-tag \#2DaysUntilMITAM, starting from 14:47:58 November 11th, 2015. The messages were collected using the Python Tweepy-package (www.tweepy.org). Only the actual message text was used (no headers), consisting of at most 140 Unicode characters. From each message a single bit was obtained by first converting each character into an integer representing its Unicode code point, converting the integer to the smallest binary bit-string representing that number and finally taking the parity of all the resulting bit-strings together (even or uneven number of ones). The dataset for B was similarly obtained from 134501 messages with the hash-tag \#3DaysTillPURPOSE, streamed prior to the dataset A, starting from 16:52:44 November 10th, 2015.
\begin{figure}[hbt]
	\centering
	\includegraphics[width=83mm]{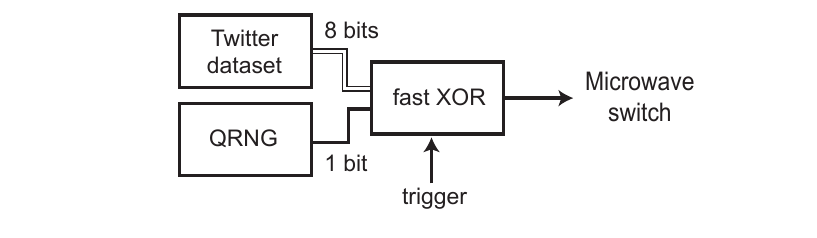}
	\caption{ \label{fig:bell_rnd_twitter} Schematic of random input bit generation by combining bits from a quantum random number generator (QRNG) and classical random bits from a dataset based on Twitter messages.}
\end{figure}

We note that although one may expect the Justin~Bieber and One~Direction fan-bases to be sufficiently disjoint to produce an uncorrelated binary dataset, the hashtag from dataset B featured in 2 out of 139952 tweets of dataset A, and vice-versa in 4 out of 134501 tweets. Still, a Fisher-exact independence test of A (first 134501 bits) and B's dataset results in a \pvalue\ of 0.63. The bias of the 8-bit parity sets were $0.44\%$ and $0.95\%$, with statistical uncertainty ($\frac{1}{2\sqrt{N}}$) of $0.38\%$ and $0.39\%$ for A,B respectively. As these bits are XOR'ed with bits from the quantum random number generator with much smaller bias, this has no expected effect on the bias in the used input settings. Finally, we characterized the performance of the FPGA combiners, which showed no errors on $10^8$ XOR operations.

\subsection*{APD replacement}
\label{sec:bell_apd_replacement}
After 5 days of measurement, the APD at location C corresponding to channel~0 broke down during the daily calibration routine and was subsequently replaced. To take into account the changed detection-to-output delay for the event-ready filter settings, the laser pulse arrival time was recorded for the new APD before proceeding. We adapted the start of the event-ready window for channel~0 accordingly, and used this for all the data taken afterwards.

\subsection*{Joint \pvalue~for $\psi^-$ and $\psi^+$ heralded events}
\label{sec:bell_joint_pvalue_psi_plus}
Here we expand the statistical analysis used for the first run\cite{Hensen2015LoopholeFree} to incorporate the $\psi^-$ and $\psi^+$ events into one hypothesis test. For each of these states we perform a different variant of the CHSH game, and then use the methods of Elkouss and Wehner\cite{Elkouss2015Nearly} to combine the two: The output signal of the ``event-ready''-box  $\mathbf{\optTilde{t}}^{\NNN} = (\optTilde{t}_\ellInd)_{\ellInd=1}^\NNN$ now has three possible outcomes, where the tag $\optTilde{t}_\ellInd=0$ still corresponds to a failure (no, not ready) event. We now distinguish two different successful preparations of the boxes A and B: $\optTilde{t}_\ellInd=-1$ corresponds to a successful preparations of the $\psi^-$ Bell state, and $\optTilde{t}_\ellInd=+1$ to a $\psi^+$ Bell state. In terms of non-local games, Alice and Bob are playing two different games, where in case $\optTilde{t}_\ellInd=-1$ they must have $(-1)^{\optTilde{a}_{\ellInd} \optTilde{b}_{\ellInd}} \optTilde{x}_{\ellInd}\optTilde{y}_{\ellInd} = 1$ in order to win, and in case of  $\optTilde{t}_\ellInd=+1$ they must have $(-1)^{\optTilde{a}_{\ellInd} (\optTilde{b}_{\ellInd}\oplus1)} \optTilde{x}_{\ellInd}\optTilde{y}_{\ellInd} = 1$ to win. Note that both games have the same maximum winning probabilities. This means that we can take $k:= k_- + k_+$, with $k_-$ the number of times $(-1)^{\optTilde{a}_{\ellInd} \optTilde{b}_{\ellInd}} \optTilde{x}_{\ellInd}\optTilde{y}_{\ellInd} = 1$, and $k_+$ the number of times $(-1)^{\optTilde{a}_{\ellInd} (\optTilde{b}_{\ellInd}\oplus1)} \optTilde{x}_{\ellInd}\optTilde{y}_{\ellInd} = 1$; the remainder of the analysis remains the same and in particular the obtained bound to the \pvalue\ is unchanged (see Elkouss and Wehner\cite{Elkouss2015Nearly}, page 20).
We then have for the adapted CHSH function (see Supplementary Information of \firstrunpaper):
\begin{equation}
\label{eq:bell_total_k_adapted}
k' :=
\sum_{\ellInd=1}^\NNN |\optTilde{t}_\ellInd| \cdot \frac{(-1)^{\optTilde{a}_\ellInd (\optTilde{b}_\ellInd+\frac{\optTilde{t}_\ellInd+1}{2})}\optTilde{x}_\ellInd \optTilde{y}_\ellInd+1}{2}\ ,
\end{equation}
and adapted total number of events then becomes: 
\begin{align}
\label{eq:bell_total_n_adapted}
n' &:= |\optTilde{\mathbf{t}}^{\NNN}| = \sum_{\ellInd=1}^\NNN |\optTilde{t}_\ellInd|.
\end{align}

\subsection*{Results}
In this test we set the total number of Bell trials $n_2=300$. After 210 hours of measurement over 22 days during 1 month, we find $S_2=2.35 \pm 0.18$, with $S_2$ the weighted average of $S_{\psi^-} =\left< x \cdot y \right>_{\left(0,0\right)} + \left< x \cdot y \right>_{\left(0,1\right)} 
    + \left< x \cdot y \right>_{\left(1,0\right)} - \left< x \cdot y \right>_{\left(1,1\right)}$ for $\psi^-$ heralded events (different detectors clicked), and $S_{\psi^+} =\left< x \cdot y \right>_{\left(0,0\right)} + \left< x \cdot y \right>_{\left(0,1\right)} 
    - \left< x \cdot y \right>_{\left(1,0\right)} + \left< x \cdot y \right>_{\left(1,1\right)}$ for $\psi^+$ (same detector clicked). See Figure \ref{fig:bell_new_data}. 
    
This yields a \pvalue\ of 0.029 in the conventional analysis~\cite{Hensen2015LoopholeFree} (which assumes independent trials, perfect random number generators and Gaussian statistics), and with $k_2 = 237$ a \pvalue\ of 0.061 in the complete analysis~\cite{Hensen2015LoopholeFree} (which allows for arbitrary memory between the trial, partially predictable random inputs and makes no assumptions about the probability distributions).

\begin{figure}[htb]
	\centering
	\includegraphics[width=83mm]{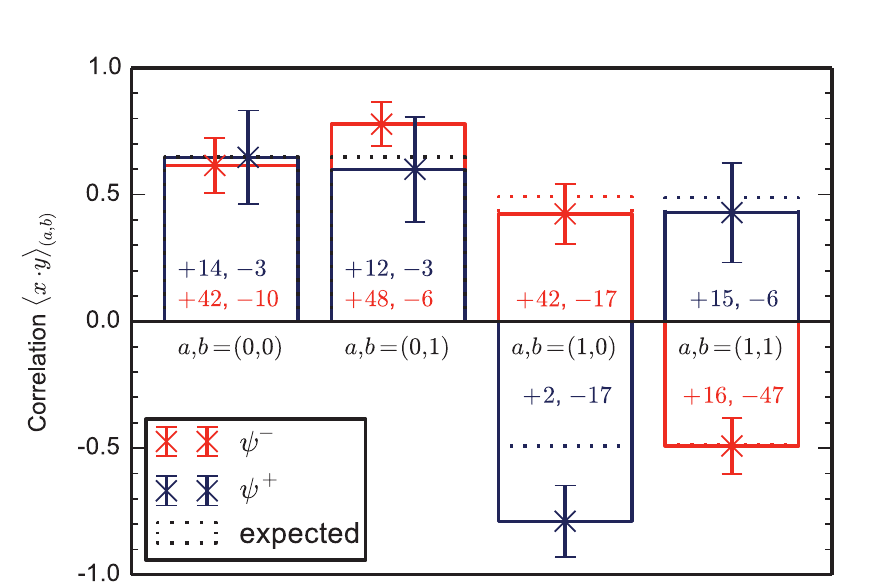}
	\caption{\label{fig:bell_new_data} \textbf{Second loophole-free Bell test results.} \textbf{(a)} Summary of the data and the CHSH correlations. We record a total of $n_2=300$ trials of the Bell test. Dotted lines indicate the expected correlation based on the spin readout fidelities and the characterization measurements presented in \firstrunpaper. Shown are data for both $\psi^-$ heralded events (red, two clicks in different APD's at location C), and for $\psi^+$ heralded events (blue, two clicks in the same APD). Numbers in bars represent the amount of correlated and anti-correlated outcomes respectively, for $\psi^-$ (red) and $\psi^+$ (orange). Error bars shown are $\sqrt{\left(1-\left< x \cdot y \right>_{(a,b)}^2\right)/n_{(a,b)}}$, with $n_{(a,b)}$ the number of events with inputs $(a,b)$. }
\end{figure}

\section*{Combined \pvalue~for the two tests}
\label{sec:combining_datasets}
We now turn to analysing the statistical significance of the two runs combined. Extending the conventional analysis, we take the weighted sum of the CHSH parameters obtained for both tests to find $S_{\mathrm{combined}}=2.38\pm0.136$, yielding a \pvalue\ of 2.6$\cdot 10^{-3}$.
For the complete analysis, we give here two cases. The first case is where the tests are considered to be fully independent; we can then combine the \pvalues\ using Fisher's method, resulting in a joint \pvalue\ of $1.7 \cdot 10^{-2}$  for the complete analysis. In the second case we consider the two runs as a single test; we can then combine the data, and find $k_1+k_2= 433$ for $n_1+n_2= 545$, resulting in a joint \pvalue\ of 8.0$\cdot 10^{-3}$ for the complete analysis. We emphasize that these are extreme interpretations of a subtle situation and these \pvalues\ should be considered accordingly.

Although the predefined event-ready filter settings were used for the hypothesis tests presented, the datasets recorded during the Bell experiments contain all the photon detection times at location C. This allows us to investigate the effect of choosing different heralding windows in post-processing. Such an analysis does not yield reliable global \pvalues\ (look-elswehere effect), but can give insight in the physics and optimal parameters of the experiment. In Fig.~\ref{fig:bell_filter_dependence} we present the dependence of the recorded Bell violation $S$, and number of Bell trials $n$, if we offset the start of the windows. For negative offsets, photo-detection events caused by reflected laser light starts to play an important role, and as expected the Bell violation decreases since the event-ready signal is in that regime no longer a reliable indicator of the generation of an entangled state. The observed difference between the runs in offset times at which the laser reflections start to play a role are caused by the less aggressive filter settings in the second run. However, we see that in both runs the $S$-value remains constant up a negative offset of about 0.8 ns, indicating that the filter settings were still chosen on the conservative side.

\begin{figure*}[tb]
	\centering
	\includegraphics[width=134mm]{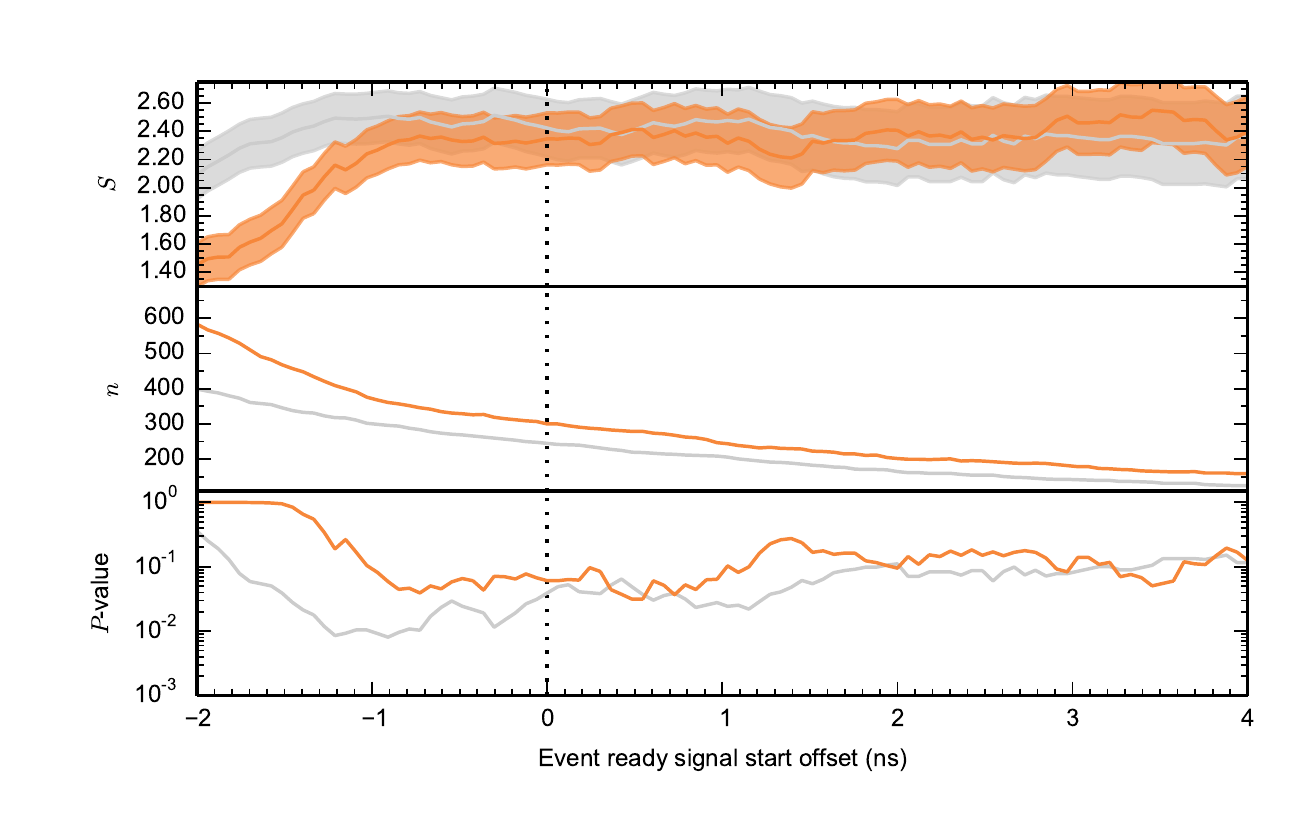}
	\caption{ \label{fig:bell_filter_dependence}
	CHSH parameter $S$, number of Bell trials $n$, and post-selected complete-analysis local \pvalue\ versus window start offset for the event-ready photon detections at location C, for the first (grey) and second (orange) dataset.
	The time-offset shown is with respect to the predefined windows (corresponding to the dotted line). Confidence region shown is one sigma, calculated according to the conventional analysis. Shifting the window back in time, the relative fraction of heralding events caused by photo-detection from laser reflections increases, thereby reducing the observed Bell violation.}
\end{figure*}

\section*{Statistical analysis of settings choices}
\label{sec:settings_choice_analysis}
Both for the Bell run in \firstrunpaper\ and for the Bell run presented above, we are testing a single well-defined null hypothesis formulated before the experiment, namely that a local-realist model for space-like separated sites could produce data with a violation at least as large as we observe. The settings independence is guaranteed by the space-like separation of relevant events (at stations A, B and C). Since no-signalling is part of this local-realist model, there is no extra assumption that needs to be checked in the data. We have carefully calibrated and checked all timings to ensure that the locality loophole is indeed closed.

Nonetheless, one can still check (post-experiment) for many other types of potential correlations in the recorded dataset if one wishes to. However, since now many hypotheses are tested in parallel, \pvalues\ should take into account the fact that one is doing multiple comparisons (the look-elsewhere effect, LEE). Failure to do so can lead to too many false positives, an effect well known in particle physics. In contrast, there is no LEE for a single pre-defined null hypothesis as in our Bell test.

Formulation and testing of multiple hypotheses can result in obtaining almost arbitrarily low local  \pvalues, which may have almost no global significance. As an example, recalculating the \pvalue\ for the local realist hypothesis, given the first dataset for a window start offset of -900 picoseconds compared to the predefined window starts, results in a local \pvalue\ of 0.0081 using the complete analysis (see Fig.~\ref{fig:bell_filter_dependence}). Taking this to the extreme by doing a search of the window start offsets for both channels independently and the joint window length offset, results in a local \pvalue\ of 0.0018. These examples clearly illustrate that without taking into account that multiple hypotheses are being tested, such local \pvalues\ can not be used to assign significance.

With these considerations in mind we analyse the settings choices in the two sub-sections below.

\subsection*{Settings choices in the first and second dataset}

The distribution of the $245$ input settings in the first dataset (see Figure 4a in \firstrunpaper) is $(n_{(0,0)},n_{(0,1)},n_{(1,0)},n_{(1,1)}) = (53,79,62,51)$, with $n_{(a,b)}$ the number times the inputs $(a,b)$ were used. This realisation looks somewhat unbalanced for a uniform distribution, and one could be motivated to test the null hypothesis that the RNGs are uniform. Performing a Monte-Carlo simulation of $10^5$ realisations of a uniform multinomial distribution with size $n=245$ we find a local \pvalue\ of 0.053 to get such a distribution or more extreme. We can get further insight by looking at all the setting choices recorded during the test. Around every potential heralding event about 5000 settings are recorded, for which we find a local \pvalue\ of 0.57 (Table \ref{tbl:pvalues_settings}), consistent with a uniform setting distribution. 

Many additional tests can be performed on equally many slices or subsets of the data, where one or more of the filters (see Supplementary Information of \firstrunpaper) is relaxed.
In Table~\ref{tbl:pvalues_settings} we list the individual (local) \pvalues\ for a set of 4 hypotheses regarding the settings choices, for both the first and second dataset.
\begin{enumerate}
\item RNG A is uniform 
\item RNG B is uniform
\item RNG A and RNG B are jointly uniform
\item Fisher's exact test~\cite{Fisher1922On} for $n<5000$, Pearson's $\chi^2$ test~\cite{Pearson1900X} for $n>5000$)
\end{enumerate}
For tests 1 and 2 we evaluate a two tailed binomial test with equal success probability. For test 3 we perform a Monte-Carlo simulation of $10^5$ realisations of a uniform multinomial distribution with size fixed to the number of observations in that particular row, i.e. $n=245$ for the second row in Table~\ref{tbl:pvalues_settings}.

\begin{table*}
	\centering
	\includegraphics[width=\textwidth]{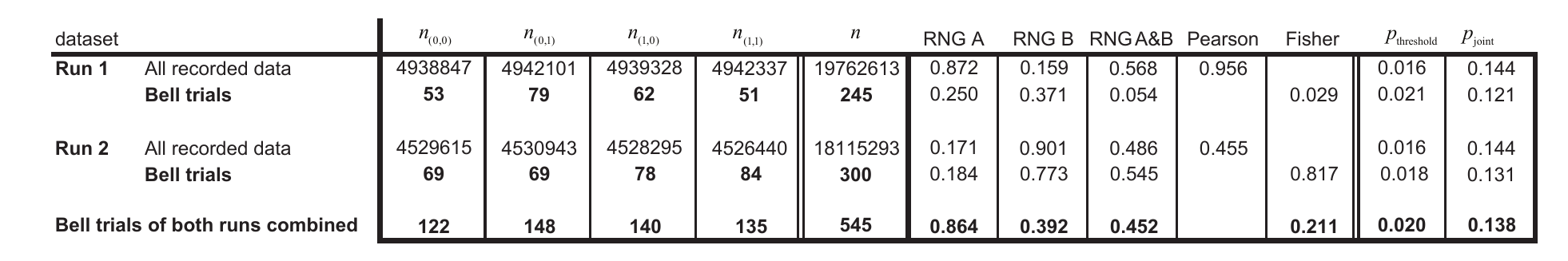}

  \caption{\label{tbl:pvalues_settings}From left to right each column corresponds with: dataset on which statistics are computed, local \pvalue\ for the null hypothesis RNG A is uniform, local \pvalue\ for the null hypothesis RNG B is uniform, local \pvalue\ for the null hypothesis RNG A\&B is uniform, Fisher's test, Pearson's test, and $p_{\textrm{threshold}}$ and joint \pvalues $p_{\textrm{joint}}$. The joint \pvalue\ for a set of hypotheses is the probability that for at least one of the hypotheses we observe a \pvalue\ less than $\alpha$ where here $\alpha = 0.05$. This captures the fact that the more hypotheses we test, the more likely it becomes that one of them will fall below the significance threshold. The value $p_{\textrm{threshold}}$ the largest threshold for individual tests for which the joint \pvalue\ for that row is less than 0.05. The local \pvalues\ in the row should this be compared to this number. This captures the fact that when testing multiple hypothesis, the local \pvalues\ of the individual ones actually need to be much smaller for the overall test to be significant. The local \pvalues\ in columns RNG A, RNG B, Fisher and Pearson are exact calculations. The columns RNG A\&B, $p_{\textrm{threshold}}$ and $p_{\textrm{joint}}$ are approximations obtained via $10^5$, $10^4$ and $10^4$ trials of a Monte-Carlo simulation, respectively.}
\end{table*}

We observe that only one local \pvalue\ is below 0.05: Fisher's exact test on the distribution of the settings in the first data set yields a local \pvalue\ of 0.029. However, as described in the next subsection below, when properly taking the look-elsewhere effect into account this does not result in a significant rejection of the uniform settings hypothesis at the 0.05 level. Finally, the valid Bell trials of the first and second dataset combined, shown in the last row of Table~\ref{tbl:pvalues_settings}, are also consistent with uniformly chosen input settings.

\subsection*{Significance and look-elsewhere effect}
\label{sec:LEE}
We now analyse the significance of the local \pvalues\ in Table~\ref{tbl:pvalues_settings} by taking into account the look-elsewhere effect. Say we are looking for correlations between parameters that are in fact completely independent. Looking at one correlation, it is as if we take one random sample from a distribution; the probability that it is at 2 sigma or more extreme is thus about 0.05. If we look for 4 different correlations (assuming all parameters are independent), it is similar to taking 4 random samples, and thus the probability that at least one is at 2 sigma or more extreme is $1-(1-0.05)^4=0.18$. In reverse, assuming fully independent hypotheses, the local \pvalue~$p'$ should have obeyed roughly $1-(1-p')^4<0.05$, so $p'< 0.013$, to be statistically significant at the 0.05 level.

In our case it is actually more complicated because there can be dependencies between hypotheses. We can numerically get some of these numbers. For instance, we have simulated the random number generation (RNG) using Monte-Carlo under the assumption of independent uniform outputs and calculated local \pvalues\ for the four hypotheses listed above. The probability that at least one of these yield local \pvalue~$p' < 0.05$ turns out to be about 0.13 for the 245 events in the Bell test. This is different from $1-(1-0.05)^4=0.18$ because of correlations between the tests, but it is clearly much higher than 0.05. In reverse, to arrive at an overall probability of 0.05 of finding at least one test yielding local \pvalue~$p' < p_\mathrm{threshold}$ for the data in the first Bell dataset, we find $p_\mathrm{threshold}=0.021$. In other words, if we would only be looking at the settings corresponding to the valid Bell trials, then a local \pvalue\ below 0.021 would signal a statistically significant violation of our hypothesis at the 0.05 level. We do not find such evidence for the valid Bell trial data (see first row in Table~\ref{tbl:pvalues_settings}).

The last column gives the probability that at least one of the hypothesis tests on the data in that row yields a local \pvalue~$p'< 0.05$, given uniform settings. In the one-but-last column we give $p_\mathrm{threshold}$, again only for the data set in that row, for a significance at the 0.05 level. These values assume that we would only be testing our hypotheses on that particular row. Since we are now looking at different rows, $p_\mathrm{threshold}$ for each row is a strict upper bound to the $p_\mathrm{threshold}$ for the full table, as we are looking at different cross-sections of the raw data set at the same time; the $p_\mathrm{threshold}$ for the full table will thus be lower but it is not trivial to compute this, given the large dependence between the subsets of data used for each row. However, since we do not find any local \pvalue\ to be below $p_\mathrm{threshold}$ for the corresponding row, we can conclude that the data does not allow rejection of the settings independence hypothesis, even without calculating the global $p_\mathrm{threshold}$ for the full Table.

\section*{Refined analysis of imperfect random number generators}
\label{sec:refined_lhvms}
Ideally, the RNGs yield a fully unpredictable random bit in every trial of the Bell test. A deviation from the ideal behaviour can be denoted by an excess predictability or bias $b$, that can take on values between 0 and $\frac{1}{2}$. In principle the value of $b$ can be different in every trial of a Bell test, which can be modelled by some probability distribution over the value of $b$. By characterising the physical RNGs, we can hope to learn something about the \emph{mean} $\tau$ of this probability distribution. As a particular example of an underlying probability distribution for the bias, consider the case where the random bit is perfectly predictable ($b=\frac{1}{2}$) with probability $f$ and perfectly unpredictable ($b=0$) with probability $1-f$ . This example could model a scenario where the random numbers are generated with some spread in time such that some of them are produced so early that they could be known by the other party before the end of the trial.

A recent analysis of the effect of partial predictability of RNGs on the bound of the CH-Eberhard inequality revealed a strong dependence on the interpretation of the mean excess predictability\cite{Kofler2014Requirements}, estimated from characterisation of the RNGs. In particular, for a model in which the mean excess predictability $\epsilon$ is distributed (evenly) over all trials, the CH-Eberhard inequality can be violated even if the relevant Bell parameter $J$ (which can be viewed as an average violation per trial in terms of probabilities) is much lower than $\epsilon$. On the other hand, Kofler \emph{et. al.}\cite{Kofler2014Requirements} found that in case of an all-or-nothing scenario, such that in a fraction $\epsilon$ of the trials the RNG is fully predictable and in the rest of the trials fully unpredictable, the threshold value for a violation is roughly given by $J>\epsilon$.

Motivated by these findings, we generalize here the analysis of the effect of imperfect random number generators on the winning probability per trial in the CHSH game. We extend the analysis in the Supplementary Information of \firstrunpaper\ (see also Elkouss and Wehner\cite{Elkouss2015Nearly}) to the case where any bias $b$ is produced by an arbitrary underlying probability distribution per trial. That is, there is no maximum bias, but rather the bias can probabilistically take on any value. We find that in our case, as long as the event-ready signal is independent of the random bits, the only relevant parameter is the mean $\tau$ of the bias; the concrete form of the random variable has no impact on the bound on the probability of winning CHSH. 

In the analysis below we explicitly take into account the possibility of early production of random bits, which we expect to be a particular interpretation of the probability distribution over $b$ as above. Indeed we find that when the random bits are perfectly predictable with probability $f$, and perfectly unpredictable with probability $1-f$, then a distribution over the bias $b$ with a mean of $\tau=f/2$ links the two viewpoints of the analysis.

\begin{figure}
\centering
\includegraphics[width=83mm]{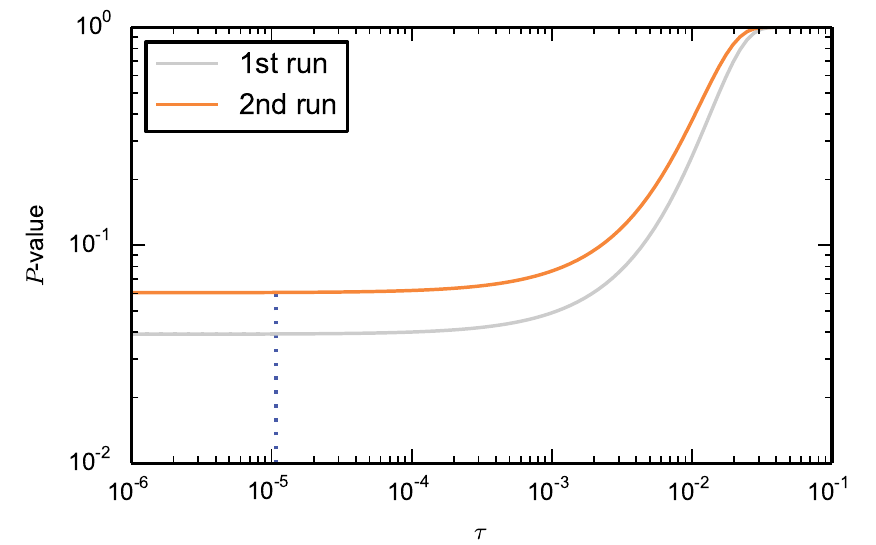}
\caption{\label{fig:fchsh} The \pvalue\ of the two runs as a function of $\tau$, the mean bias of the RNG.}
\end{figure}
 
In order to make the discussion precise, in the following we describe the random variables that characterize the experiment, then make a rigorous derivation of the winning probability.

\subsection*{Properties of the tested \LHVMs}
\label{sec:lhvmdesc}
We introduce the following sequences of random variables. The notation and arguments borrow from earlier work\cite{Brunner2014Bell,Bierhorst2014Rigorous,Bierhorst2015Robust,Elkouss2015Nearly}. Let  
$\mathbf{A}^{\N} = (A_\ellInd)_{\ellInd=1}^\N,\mathbf{B}^{\N} = (B_\ellInd)_{\ellInd=1}^\N$ the outputs of the boxes where $\ellInd$ is used to label the $\ellInd$-th element,  
 $\mathbf{\HHTot}^{\N} = (\HHTot_\ellInd)_{\ellInd=1}^\N$ the histories of attempts previous to the $\ellInd$-th attempt, $\mathbf{C}^{\N} = (C_\ellInd)_{\ellInd=1}^\N$ denotes the scores at each \TT\ and $\mathbf{T}^{\N} = (T_\ellInd)_{\ellInd = 1}^{\N}$ is the sequence of event-ready signals in the case of an event-ready experiment.  In an event-ready experiment, we make no assumptions regarding the statistics of the event-ready station, which may be under full control of the \LHVM, and can depend arbitrarily on the history of the experiment. 

We introduce three sequences of random variables to model each RNG. Let $\mathbf{X}^{\N} = (X_\ellInd)_{\ellInd=1}^\N,\mathbf{Y}^{\N} = (Y_\ellInd)_{\ellInd=1}^\N$ denote the inputs to the boxes. Let $\mathbf{\Fx}^{\N} = (\Fx_\ellInd)_{\ellInd=1}^\N,\mathbf{\Fy}^{\N} = (\Fy_\ellInd)_{\ellInd=1}^\N$ denote two sequences of binary variables that take value 1 if the random number was generated so early that signaling is possible and 0 otherwise. We call the former an early number and the latter a on-time number. Finally, let  
$\mathbf{\Bx}^{\N} = (\Bx_\ellInd)_{\ellInd=1}^\N,\mathbf{\By}^{\N} = (\By_\ellInd)_{\ellInd=1}^\N$ take values in the range $[-1/2,1/2]$ and denote the bias of the random number generators at each attempt. We here assume that these distributions can differ for all $\ellInd$, they do not depend on the history $\HHTot_\ellInd$. Using more involved notation, the same bound can be made if their mean is known conditioned on the history.

The random variable $\HHTot_\ellInd$ models the state of the experiment prior to the measurement. As such, $\HHTot_\ellInd$ includes 
any hidden variables, sometimes denoted using the letter $\lambda$ \cite{Brunner2014Bell}. It also includes 
the history of all possible configurations of inputs and outputs of the prior \TTs\ $(X_j, Y_j, A_j, B_j, T_j)_{j=1}^{\ellInd-1}$.

The null hypothesis (to be refuted) is that our experimental setup can be modelled using a \LHVM. LHVMs verify the following conditions:

\begin{enumerate}
\item \emph{Independent random number generators.}
Conditioned on the history of the experiment the random numbers are independent of each other
\begin{align}
\forall \ellInd, X_\ellInd, \Bx_\ellInd, \Fx_\ellInd \indep Y_\ellInd,\By_\ellInd,\Fy_\ellInd\mid H_\ellInd\ ,
\end{align}
and of the output of the event-ready signal
\begin{equation}
\forall \ellInd, X_\ellInd, \Bx_\ellInd, \Fx_\ellInd \indep T_\ellInd, Y_\ellInd,\By_\ellInd,\Fy_\ellInd \indep T_\ellInd \mid H_\ellInd\ .
\end{equation}
We allow $X_\ellInd$ and $Y_\ellInd$ to be partially predictable given the history of the experiment. The predictability is governed by some random variables $\Bx_\ellInd,\By_\ellInd$. For $b_x\in\Bx_\ellInd,b_y\in \By_\ellInd$ we have 
\begin{align}
\forall (\ellInd,x_\ellInd,h_\ellInd), \frac{1}{2} - b_x \leq \pr{X_\ellInd = x_\ellInd|H_\ellInd = h_\ellInd,\Bx_\ellInd=b_x,\By_\ellInd=b_y} &\leq \frac{1}{2} + b_x\ , \label{eq:AliceBias}\\
\forall (\ellInd,y_\ellInd,h_\ellInd), \frac{1}{2} - b_y \leq \pr{Y_\ellInd=y_\ellInd|H_\ellInd = h_\ellInd,\Bx_\ellInd=b_x,\By_\ellInd=b_y} &\leq \frac{1}{2} + b_y\ ,\label{eq:BobBias}
\end{align}
Furthermore, from the characterization of the devices we have that for all $\ellInd$: $\E{\Fx_\ellInd}\leq f_A,\E{\Bx_\ellInd}\leq\tau_A,\E{\Fy_\ellInd}\leq f_B,\E{\By_\ellInd}\leq\tau_B$. 
We define $f=\max\{f_A,f_B\}$ and $\tau = \max\{\tau_A,\tau_B\}$.
\item \emph{Locality.}
The outputs $a_i$ and $b_i$ only depend on the local input settings and history: they are independent of each other and of the input setting at the other 
side, conditioned on the previous history, the current event-ready signal and the inputs being generated on-time
\begin{equation}
\label{eq:locality}
\forall \ellInd, (X_{\ellInd},A_\ellInd) \indep (Y_{\ellInd},B_\ellInd)| \Fx_\ellInd=0,\Fy_\ellInd=0,\HH_\ellInd, T_\ellInd\ .
\end{equation}
\item \emph{Sequentiality of the experiments.}
Every one of the $\NNN$ \TTs\ takes place sequentially such that any possible signalling between different \TTs\ beyond the previous conditions is prevented \cite{Barrett2002Quantum}. 
\end{enumerate}
Except for these conditions the variables might be correlated in any possible way.

\subsection*{Winning probability for imperfect random number generators}
\label{sec:winLoseProofs}

Here, we derive a tight upper bound on the winning probability of CHSH with imperfect random number generators in an event-ready setup. 
For CHSH, the inputs $X_i,Y_i$, outputs $A_i,B_i$ and output of the heralding station $T_i$ take values 0 and 1. If $T_i=0$ the scoring variable $C_i$ takes always the value zero, if $T_i=1$ then $C_i=1$ when $x\cdot y = a\oplus b$ and $C_i=0$ in the remaining cases. We will take the RNGs to have maximum advantage $f$ of producing early random numbers. 

\begin{lemma}
\label{lem:tool}
Let $m \in \mathbb{N}$, and let a sequence of random variables as described in the previous section correspond with $m$ \TTs\ of a CHSH heralding experiment. 
Suppose that the null hypothesis holds, i.e., nature is governed by an \LHVM. Given that for all $\ellInd\leq m$: $\E{\Fx_\ellInd} = f,\E{\Bx_\ellInd}\leq\tau,\E{\Fy_\ellInd} = f,\E{\By_\ellInd}\leq\tau$, we have for $i\leq m$, any possible history $\HH_\ellInd =\h_\ellInd$ of the experiment, and $T_\ellInd = 1$ that the probability of $C_\ellInd=1$ is upper bounded by
\begin{equation}
\pr{C_\ellInd=1|\HH_\ellInd = \h_\ellInd, T_\ellInd = 1}\leq \bellwin^1\ ,
\end{equation}
where $\bellwin^1=2f-f^2+(1-f)^2 \left(\frac{3}{4}+\tau'-\tau'^2 \right)$ and  $\tau':=\min\{\frac{2\tau+f}{2(1-f)},\frac{1}{2}\}$.
\end{lemma}
\begin{proof}
Let us first bound the effect of the early random numbers in the winning probability. We have
\begin{align}
\pr{C_\ellInd=1|\HH_\ellInd = \h_\ellInd,T_\ellInd=1} &= \sum_{s_x,s_y\in\{0,1\}}\pr{\Fx_\ellInd=s_x,\Fy_\ellInd=s_y}\pr{C_\ellInd=1|\HH_\ellInd = \h_\ellInd,T_\ellInd=1,\Fx_\ellInd=s_x,\Fy_\ellInd=s_y} \\
&\leq \pr{\Fx_\ellInd=0,\Fy_\ellInd=1}+\pr{\Fx_\ellInd=1,\Fy_\ellInd=0}
+\pr{\Fx_\ellInd=1,\Fy_\ellInd=1}\nonumber\\
&\quad +\pr{\Fx_\ellInd=0,\Fy_\ellInd=0}\pr{C_\ellInd=1|\HH_\ellInd = \h_\ellInd,T_\ellInd=1,\Fx_\ellInd=s_x,\Fy_\ellInd=s_y}\\
&\leq 2f-f^2+(1-f)^2\pr{C_\ellInd=1|\HH_\ellInd = \h_\ellInd,T_\ellInd=1,\Fx=0,\Fy=0}
\end{align}
The first inequality follows by assuming that the CHSH is  won with probability one when a random number is early. The second inequality follows from assuming that $\E{\Fx_\ellInd}=\E{\Fy_\ellInd}= f$.

Let us now bound the bias for the on-time numbers.
We focus on the random numbers $X_\ellInd$; the same argument can be made for $Y_\ellInd$. For simplicity, we omit the explict conditioning on $\HH_\ellInd 
= \h_\ellInd$.
First of all, note that since
\begin{equation}
\pr{X_\ellInd=1} = \int db_x \pr{\Bx_\ellInd = b_x} 
\pr{X_\ellInd=1| \Bx_\ellInd=b_x}
\end{equation}
we have together with~\eqref{eq:AliceBias} that
\begin{equation}
\frac{1}{2}-\mathbb{E}[\Bx_\ellInd] \leq \pr{X_\ellInd=1}\leq \frac{1}{2}+\mathbb{E}[\Bx_\ellInd]\label{eq:bound2}
\end{equation}
which implies $1/2-\tau\leq \pr{X_\ellInd=1}\leq 1/2+\tau$. 
Furthermore, note that we can expand the probability as
\begin{equation}
\pr{X_\ellInd=1}=\pr{\Fx_\ellInd=1}\pr{X_\ellInd=1|\Fx_\ellInd=1}+\pr{\Fx_\ellInd=0}\pr{X_\ellInd=1|\Fx_\ellInd=0}\label{eq:bound1}\ .
\end{equation}
Combining \eqref{eq:bound1}, \eqref{eq:bound2} and the assumption
$\mathbb{E}(\Fx_\ellInd) = f$ we obtain
\begin{equation}
 \pr{X_\ellInd=1|\Fx_\ellInd=0}\leq \frac{1}{2}+\tau'
\end{equation}
where $\tau':=\min\{\frac{2\tau+f}{2(1-f)},\frac{1}{2}\}$.
Let us now expand the probability that $C_\ellInd=1$ conditioned on the event that both numbers were on-time. For simplicity, we drop 
the explicit conditioning on $\HH_\ellInd = \h_\ellInd,T_\ellInd=1,\Fx=0,\Fy=0$. 
\begin{align}
\pr{C_\ellInd=1} &= \sum_{\substack{{x,y,z\in\{0,1\}}\\(x,y)\neq(1,1)}}
\pr{A_\ellInd=z,B_\ellInd=z,X_\ellInd=x,Y_\ellInd=y}+\sum_{\substack{z\in\{0,1\}}}\pr{A_\ellInd=z,B_\ellInd=z\oplus 1,X_\ellInd=1,Y_\ellInd=1}\ .\label{eq:probexpansion1}
\end{align}
We can break these probabilities into simpler terms
\begin{align}
\p\big(A_\ellInd=a,B_\ellInd=b,X_\ellInd=x,Y_\ellInd=y) &=\pr{A_\ellInd=a,X_\ellInd=x}\pr{B_\ellInd=b,Y_\ellInd=y}\\
&=\pr{X_\ellInd=x}\pr{A_\ellInd=a|X_\ellInd=x} \pr{Y_\ellInd=y} \pr{B_\ellInd=b|Y_\ellInd=y}\ . 
\end{align}
The first equality followed by the locality condition, the second one simply by the definition of conditional probability.
With this decomposition, we can express \eqref{eq:probexpansion1} as
\begin{align}
\pr{C_\ellInd=1} &= \sum_{\substack{{x,y\in\{0,1\}}\\(x,y)\neq(1,1)}}\alpha_{x}\beta_{y}\left(\chi_{x}\gamma_{y}+(1-\chi_{x})(1-\gamma_{y})\right)+\alpha_{1}\beta_{1}\left(\chi_{1}(1-\gamma_{1})+(1-\chi_{1})\gamma_{1}\right)\\
                                                                &= \sum_{x,y\in\{0,1\}} \alpha_{x}\beta_{y}f_{x,y}\ . \label{eq:compactsum}
\end{align}
where we have used the shorthands
\begin{align}
\chi_{x}&:=\pr{A_\ellInd=1}\ , \\
\gamma_{y}&:=\pr{B_\ellInd=1},\\
\alpha_{x}&:=\pr{X_\ellInd=x}\ ,\\
\beta_{y'}&:=\pr{Y_\ellInd=y}\ ,\\
f_{x,y} &:= 
\left\{ 
\begin{aligned}
\chi_{x}\gamma_{y}+(1-\chi_{x})(1-\gamma_{y}) & \qquad\textrm{ if } (x,y)\neq (1,1)\ , \\
\chi_{x}(1-\gamma_{y})+(1-\chi_{x})\gamma_{y} & \qquad\textrm{ otherwise.}
\end{aligned}
\right.
\end{align}
Now we will expand \eqref{eq:compactsum}. We know that $1/2 -\tau\leq\alpha_{x},\beta_{y}\leq 1/2 +\tau$. In principle, $\tau$ does not need to take the values in the extreme on the range. Without loss of generality let $\alpha_{0}=1/2+\tau_A$ and $\beta_0=1/2+\tau_B$, with $\tau_A,\tau_B\in[-1/2,1/2]$.
\begin{align}                                                                
\sum_{x,y\in\{0,1\}} \alpha_{x}\beta_{y}f_{x,y}  &= \left(\frac{1}{2}+\tau_{A}\right)\left(\frac{1}{2}+\tau_{B}\right)f_{0,0} + \left(\frac{1}{2}+\tau_{A}\right)\left(\frac{1}{2}-\tau_{B}\right)f_{0,1}\nonumber\\
                                                                &\qquad+\left(\frac{1}{2}-\tau_{A}\right)\left(\frac{1}{2}+\tau_{B}\right)f_{1,0}+\left(\frac{1}{2}-\tau_{A}\right)\left(\frac{1}{2}-\tau_{B}\right)f_{1,1}\\
                                                                &= \left(\frac{1}{4}+\frac{1}{2}\tau_{A}+\frac{1}{2}\tau_{B}+\tau_{A}\tau_{B}\right)f_{0,0} + \left(\frac{1}{4}+\frac{1}{2}\tau_{A}-\frac{1}{2}\tau_{B}-\tau_{A}\tau_{B}\right)f_{0,1}\nonumber\\
                                                                &\qquad+\left(\frac{1}{4}-\frac{1}{2}\tau_{A}+\frac{1}{2}\tau_{B}-\tau_{A}\tau_{B}\right)f_{1,0}+ \left(\frac{1}{4}-\frac{1}{2}\tau_{A}-\frac{1}{2}\tau_{B}+\tau_{A}\tau_{B}\right)f_{1,1}\\
&= \left(\tau_{A}+\tau_{B}\right)f_{0,0} + \left(\tau_{A}-2\tau_{A}\tau_{B}\right)f_{0,1}\nonumber\\
                                                                &\qquad+\left(\tau_{B}-2\tau_{A}\tau_{B}\right)f_{1,0}+ \left(\frac{1}{4}-\frac{1}{2}\tau_{A}-\frac{1}{2}\tau_{B} +\tau_{A}\tau_{B}\right)\sum_{a,b}f_{a,b}\ .\label{eq:FRemain}
\end{align}                                                                
It thus remains to bound the sum of $f_{x,y}$. Note that we can write
\begin{align}
\sum_{x,y\in\{0,1\}} f_{x,y} &= \left(\chi_{0}\gamma_{0}+(1-\chi_{0})(1-\gamma_{0})\right)
                                                               + \left(\chi_{0}\gamma_{1}+(1-\chi_{0})(1-\gamma_{1})\right) \nonumber\\
                                                               &\qquad + \left(\chi_{1}\gamma_{0}+(1-\chi_{1})(1-\gamma_{0})\right)
                                                               + \left(\chi_{1}(1-\gamma_{1})+(1-\chi_{1})\gamma_{1}\right) \\
&= \chi_0 \left(\gamma_0 + \gamma_1\right) + \left(1-\chi_0\right) \left(2 - \gamma_0 - \gamma_1\right)\nonumber\\
&\qquad + \chi_1 \left(\gamma_0 + 1 - \gamma_1\right) + \left(1- \chi_1\right)\left(1 - \gamma_0 + \gamma_1\right)\label{eq:updelta}\ . 
\end{align}
Since \eqref{eq:updelta} is a sum of two convex combinations, it must take its maximum value at one of the extreme points, that is with $\chi_{0}\in\{0,1\}$ and $\chi_{1}\in\{0,1\}$. We can thus consider all four combinations of values for $\chi_0$ and $\chi_1$ given by
\begin{align}
\sum_{x,y\in\{0,1\}} f_{x,y} =
\left\{
\begin{aligned}
3-2\gamma_{0} &\textrm{ if }(\chi_{0},\chi_{1})=(0,0)\ ,\\ 
3-2\gamma_{1} &\textrm{ if }(\chi_{0},\chi_{1})=(0,1)\ ,\\
1+2\gamma_{1} &\textrm{ if }(\chi_{0},\chi_{1})=(1,0)\ ,\\
1+2\gamma_{0} &\textrm{ if }(\chi_{0},\chi_{1})=(1,1)\ .
\end{aligned}
\right.
\end{align}
Since $0 \leq \gamma_0,\gamma_1 \leq 1$, we have in all cases that the sum is upper bounded by $3$. 

Now, using~\eqref{eq:FRemain} we have
\begin{align}
\pr{C_\ellInd=1} &\leq  2\left(\tau_{A}+\tau_{B}-2\tau_{A}\tau_{B}\right)+3\left(\frac{1}{4}-\frac{1}{2}\tau_{A}-\frac{1}{2}\tau_{B} +\tau_{A}\tau_{B}\right)\\
                                                                &=\frac{3}{4}+\frac{1}{2}(\tau_A+\tau_B)-\tau_A\tau_B\label{eq:lotsoftau}\\
                                                                &\leq\frac{3}{4}+\tau'-\tau'^2 
\end{align}
where in the first inequality we bound $f_{0,0},f_{0,1},f_{1,0}$ by 1. The second inequality follows since $\tau'\leq 1/2$ and for $\tau_A,\tau_B$ below $1/2$ \eqref{eq:lotsoftau} is strictly increasing both in $\tau_A$ and $\tau_B$; this implies that the maximum is found in the extreme: $\tau=\tau_A=\tau_B$.

Finally, we can plug the bound in the winning probability given that both numbers were on time into the winning probability and we obtain:
\begin{align}
\label{eqn:final_result_ftau}
\pr{C_\ellInd=1} &\leq 2f-f^2+(1-f)^2 \left(\frac{3}{4}+\tau'-\tau'^2 \right)
&=\frac{3}{4} + \frac{1}{2}f - \frac{1}{4}f^2 + \tau -\tau^2 - 2 f \tau + f^2 \tau + 2 f \tau^2 - f^2 \tau^2.
\end{align}
\end{proof}
Equation (\ref{eqn:final_result_ftau}) shows the equal footing of $\frac{1}{2}f$ and $\tau$. This result highlights the fact that early production of random numbers is just a particular distribution underlying the bias of the random number generators, where the probability $f$ of producing an early number corresponds to a mixture of completely predictable numbers ($\tau=\frac{1}{2}$) and the probability $1-f$ to unpredictable numbers.

The finding that the only relevant RNG parameter for the winning probability of the CHSH game is the mean bias makes a Bell test based on this winning probability particularly robust. In our two Bell test runs, we find a violation in terms of the CHSH winning probability of about 0.05 and 0.04 respectively, both orders of magnitude larger than the mean bias ($<10^{-4}$), and, given our theory result, independent of the underlying distribution of bias over the trials. As depicted in Fig.~\ref{fig:fchsh}, this means for instance that for a fraction of early produced random numbers of $10^{-3}$ our \pvalues\ are hardly affected. For comparison, in this scenario an experiment using the CH-Eberhard inequality would require $J > 10^{-3}$ to obtain a violation\cite{Kofler2014Requirements}, which is two orders of magnitude beyond the state of the art of photonic experiments\cite{Giustina2015SignificantLoopholeFree,Shalm2015Strong}. This difference may be traced back to the use of event-ready detectors in our experiments, which dramatically increases the fidelity of the entangled state and thus the winning probability per Bell trial. We stress that the above only holds if the event-ready signal is still independent of the early produced random bits: in case the random bits are produced so early that they are not anymore space-like separated from the event-ready signal generation and the event-ready detector could select the Bell trials based on random bits being produced too early. In our Bell tests, this means that the above robustness spans to random bits being produced a maximum of 690~ns too early. 

\section*{Conclusion}
The loophole-free violation of Bell's inequality in the second data run reported here further strengthens the rejection of a broad class of local realistic theories. We find that the data is consistent with independent setting choices, both in the first and second dataset, as well as in the combined dataset. Refined analysis of the effect of a bias in the random number generators shows that only the mean bias plays a role in the winning probability and thus in the \pvalue\ bound for our experiments. 
The large spatial separation and the strong violation in winning probability per trial of about 0.05 makes our implementation promising for future applications of non-locality for device-independent quantum key distribution~\cite{Acin2007DeviceIndependent} and randomness generation~\cite{Colbeck2007Quantum,Pironio2010Random}.


\section*{Acknowledgements}

We thank Jan-\r{A}ke Larsson, Morgan Mitchell and Krister Shalm for useful discussions.

We acknowledge support from the Dutch Organization for Fundamental Research on Matter (FOM), Dutch Technology Foundation (STW), the Netherlands Organization for Scientific Research (NWO) through a VENI grant (THT) and a VIDI grant (SW) and the European Research Council through project HYSCORE.

\section*{Author contributions}
Author Contributions B.H. and R.H. devised the experiment. B.H., N.K., A.E.D., A.R., M.S.B., R.F.L.V. and R.N.S. built and characterized the experimental set-up. K.G. compiled the random Twitter datasets. M.M. and D.J.T. grew and prepared the diamond device substrates.  M.S.B. fabricated the devices. B.H., N.K., A.E.D., A.R. and M.S.B. collected and analysed the data, with support from T.H.T. and R.H. D.E. and S.W. performed the theoretical analysis. B.H., D.E., S.W. and R.H. wrote the manuscript. R.H. supervised the project.

\clearpage

\end{document}